\documentclass{llncs}
\usepackage[latin1]{inputenc}
\usepackage{url,amsfonts,epsfig}
\usepackage{amsmath,latexsym,amssymb,txfonts}
\usepackage{blkarray}

\usepackage[table]{xcolor}
\usepackage{pgf,tikz}
\usetikzlibrary{arrows}
\usetikzlibrary{decorations.markings}
\usepackage{graphicx}
\usepackage{caption}
\usepackage{color}
\usepackage{soul}
\tikzset{->-/.style={decoration={
  markings,
  mark=at position .97 with {\arrow{>}}},postaction={decorate}}}
\usetikzlibrary{positioning}

\usepackage{color,colortbl,caption}
\usepackage{graphicx, subfig}
\usepackage{footmisc}
\usepackage{algorithm,algorithmic}

\usepackage{soul}

\newcommand{\vect}[1]{#1}
\newcommand{\DS}{\displaystyle}

\newcommand{\h}[1]{{\color{red}{#1}}}

\title{Navigating in a sea of repeats in RNA-seq without drowning}

\author{Gustavo Sacomoto\inst{1,2} \and Blerina Sinaimeri\inst{1,2}
  \and Camille Marchet\inst{1,2} \and Vincent Miele\inst{2} \and
  Marie-France Sagot\inst{1,2} \and Vincent Lacroix\inst{1,2}}

\institute{
  INRIA Rh\^one-Alpes 
  \and
  Universit\'e Lyon 1, F-69000 Lyon; CNRS, UMR5558.
}

\begin{document}

\maketitle              

\begin{abstract}
The main challenge in \emph{de novo} assembly of NGS data is certainly
to deal with repeats that are longer than the reads. This is
particularly true for RNA-seq data, since coverage information cannot
be used to flag repeated sequences, of which transposable elements are
one of the main examples. Most transcriptome assemblers are based on
de Bruijn graphs and have no clear and explicit model for repeats in
RNA-seq data, relying instead on heuristics to deal with them. The
results of this work are twofold. First, we introduce a formal model
for representing high copy number repeats in RNA-seq data and exploit
its properties for inferring a combinatorial characteristic of
repeat-associated subgraphs. We show that the problem of identifying
in a de Bruijn graph a subgraph with this characteristic is
NP-complete. In a second step, we show that in the specific case of a
local assembly of alternative splicing (AS) events, we can
\emph{implicitly} avoid such subgraphs. In particular, we designed and
implemented an algorithm to efficiently identify AS events that are
not included in repeated regions. Finally, we validate our results
using synthetic data. We also give an indication of the usefulness of
our method on real data.
\end{abstract}


\section{Introduction}
Transcriptomes can now be studied through sequencing. However, in the
absence of a reference genome, de novo assembly remains a challenging
task. The main difficulty certainly comes from the fact that
sequencing reads are short, and repeated sequences within
transcriptomes could be longer than the reads. This short read / long
repeat issue is of course not specific to transcriptome sequencing. It
is an old problem that has been around since the first algorithms for
genome assembly.  In this latter case, the problem is somehow easier
because the coverage can be used to discriminate contigs that
correspond to repeats, \textit{e.g.} using Myer's
A-statistics~\cite{Celera} or \cite{Novak10}.  In transcriptome
assembly, this idea does not apply, since the coverage of a gene does
not only reflect its copy number in the genome, but also and mostly
its expression level. Some genes are highly expressed and therefore
highly covered, while most genes are poorly expressed and therefore
poorly covered.

Initially, it was thought that repeats would not be a major issue,
since they are mostly in introns and intergenic regions. However, the
truth is that many regions which are thought to be intergenic are
transcribed~\cite{Encode12} and introns are not always spliced out yet
when mRNA is collected to be sequenced. Repeats are therefore very
present in real samples, especially transposable elements, and cause
major problems in transcriptomic assembly.

Most, if not all current short-read transcriptome assemblers are based
on de Bruijn graphs.  Among the best known are {\sc Oases}
\cite{oases}, {\sc Trinity} \cite{trinity}, and to a lesser degree
{\sc Trans-Abyss} \cite{Transabyss} and {\sc IDBA-tran}
\cite{Idbatran}. Common to all of them is the lack of a clear and
explicit model for repeats in RNA-seq data. Heuristics are thus used
to try and cope efficiently with repeats. For instance, in {\sc Oases}
short nodes are thought to correspond to repeats and are therefore not
used for assembling genes. They are added in a second step, which
hopefully causes genes sharing repeats not to be assembled together.
In {\sc Trinity}, there is no attempt to deal with repeats explicitly.
The first module of {\sc Trinity}, Inchworm, will try and assemble the
most covered contig which hopefully corresponds to the most abundant
alternative transcript. Then alternative exons are glued to this major
transcript to form a splicing graph. The last step is to enumerate all
alternative transcripts. If repeats are present, their high coverage
may be interpreted as a highly expressed link between two unrelated
transcripts.  Overall, assembled transcripts may be chimeric or spliced
into many sub-transcripts.

In the method we developed, {\sc KisSplice}, which is a local
transcriptome assembler~\cite{kissplice}, repeats may be less
problematic, since the goal is not to assemble full-length
transcripts. Instead, {\sc KisSplice} aims at finding variations in
the transcriptome (SNPs, indels and alternative splicings).  However,
we previously reported that we were not able to deal with large parts
of the de Bruijn graph containing subgraphs associated to highly
repeated sequences~\cite{kissplice}.

Here, we try and achieve two goals: (i) give a clear formalization of
the notion of repeats in RNA-seq data, and (ii) give a practical way
to enumerate bubbles that are lost because of such repeats. Recall
that we are in a \emph{de novo} context, so we assume that neither a
reference genome/transcriptome nor a database of known repeats,
e.g. \textsc{RepeatMasker}~\cite{repeatmasker}, is available.
  
In particular, we formally introduce a model for representing the
repeats and exploit its properties to infer a parameter characterizing
repeated-associated subgraphs in a de Bruijn graph. We prove its
relevance but we also show that the problem of identifying in a de
Bruijn graph a subgraph corresponding to repeats using such a
characteristic is NP-complete. Hence, a polynomial time algorithm for
repeat identification that uses such characterization is
unlikely. Finally, we show that in the specific case of a local
assembly of alternative splicing (AS) events, we can \emph{implicitly}
avoid such subgraphs. More precisely, it is possible to find the
structures (\textit{i.e.} bubbles) corresponding to AS events in a de
Bruijn graph that are not contained in a repeated-associated subgraph.
Finally, using simulated RNA-seq data we show that the new algorithm
can improve by a factor 2 the sensitivity of \textsc{KisSplice},
while also \emph{improving} the precision. We also give an
indication of the usefulness of our method in real data.


\section{Preliminaries}
Let $\Sigma$ be an alphabet of fixed size $\sigma$.  Here we always
assume $\Sigma=\{A,C,T,G\}$.  Given $s \in \Sigma^*$, let $|s|$ denote
its length, $s[i]$ the $i$th element of $s$, and $s[i,j]$ the
substring $s[i] s[i+1] \ldots s[j]$ for any $i<j \leq |s|$.

A \emph{$k$-mer} is a sequence $s \in \Sigma^k$.  Given an integer $k$
and a set $S$ of sequences each of length $n \geq k$, we define
$span(S,k)$ as the set of all distinct $k$-mers that appear as a
subsequence in $S$.

\begin{definition}
  Given a set of sequences (reads) $R \subseteq \Sigma^*$ and an integer
  $k$, we define the directed de Bruijn graph $G_k(R)=(V,A)$ where
  $V=span(R,k)$ and $A=span(R,k+1)$.
\end{definition}

Given a directed graph $G = (V,A)$ and a vertex $v \in V$, we denote
its {\em out-neighborhood} (resp. {\em in-neighborhood}) by
$N^+(v)=\{ u | (v,u) \in A \}$ (resp.  $N^-(v)=\{ u | (u,v) \in A
\}$), and its out-degree (in-degree) by $d^+(v)=|N^+(v)|$
($d^-(v)=|N^-(v)|$). A (simple) \emph{path} $\pi = s \leadsto t$ in
$G$ is a sequence of distinct vertices $v_1, \ldots, v_l$ in which
$v_0=s$, $v_l=t$ and for each $0 \leq i \leq t$, $(v_i, v_{i+1})$ is
an arc of $G$. If the graph is weighted, \textit{i.e.}  if there is a
function $w : A \rightarrow Q_{\geq 0}$ associating a weight to every
arc in the graph, then the \emph{length} of a path $\pi$ is the sum of
the weights of the traversed arcs, and is denoted by $|\pi|$.

An arc $(u,v) \in A$ is called \emph{compressible} if $d^+(u)=1$ and
$d^-(v)=1$. The intuition behind this definition comes from the fact
that every path passing through $u$ should also pass through $v$. It
should therefore be possible to ``compress'' or contract this arc
without losing any information. Note that the compressed de Bruijn
graph~\cite{trinity,oases} commonly used by transcriptomic assemblers
is obtained from a de Bruijn graph by replacing, for each compressible
arc $(u,v)$, the vertices $u,v$ by a new vertex $x$, where $N^-(x) =
N^-(u)$, $N^+(x) = N^+(v)$ and the label is the concatenation of the
$k$-mers of $u$ and $v$ without the overlapping part.  See
Fig.~\ref{fig:arc} for an example of a compressible arc in a de Bruijn
graph.

\begin{figure}
\centering
\subfloat[]{
    \resizebox{!}{1.7cm}{%
    \begin{tikzpicture}[->,>=stealth',shorten >=1pt,auto,node distance=1.5cm,
      thick,main node/.style={rectangle,draw,font=\sffamily\bfseries}]

  \node[main node] (1) {CTG};
    \node[main node] (2) [above left of=1] {ACT};
    \node[main node] (3) [below left of=1] {TCT};
     \node[main node] (4) [right=1.4cm of 1] {TGA};
     \node[main node] (5) [above right of=4] {GAT};
     \node[main node] (6) [below right of=4] {GAG};

    \path[every node/.style={font=\sffamily\small}]
      (2) edge node [left] {} (1)
      (3) edge node [left] {} (1)
      (1) edge node [left] {} (4)
      (4) edge node [left] {} (5)
      (4) edge node [left] {} (6)    
                ;

    \end{tikzpicture} 
    }
} \qquad \qquad \qquad \qquad
\subfloat[]{
  \resizebox{!}{1.7cm}{%
    \begin{tikzpicture}[->,>=stealth',shorten >=1pt,auto,node distance=1.5cm,
      thick,main node/.style={rectangle,draw,font=\sffamily\bfseries}]

  \node[main node] (1) {CTGA};
    \node[main node] (2) [above left of=1] {ACT};
    \node[main node] (3) [below left of=1] {TCT};
     \node[main node] (4) [above right of=1] {GAT};
     \node[main node] (5) [below right of=1] {GAG};

    \path[every node/.style={font=\sffamily\small}]
      (2) edge node [left] {} (1)
      (3) edge node [left] {} (1)
      (1) edge node [left] {} (4)
      (1) edge node [left] {} (5)
                ;

    \end{tikzpicture} 
    }
} 
\caption{\footnotesize{(a) The arc $(CTG, TGA)$ is the only
    compressible arc in the de Bruijn graph ($k=3$). (b) The
    corresponding compressed de Bruijn graph. }}
\label{fig:arc}
\end{figure}
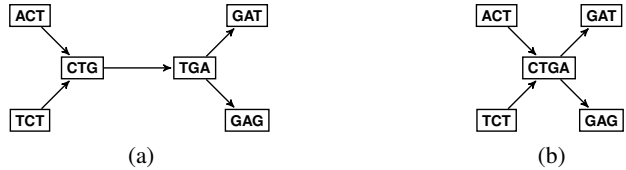

\vspace{-1cm}
\section{Repeats in de Bruijn graphs}   
Given a de Bruijn graph $G_k(R)$ generated by a set of reads $R$ for
which we do not have any information, the aim is to identify whether
there are subgraphs of $G_k(R)$ that correspond each to a set of
repeats in $R$.  To this end, we try to identify and then exploit some
of the topological properties of the subgraphs that are induced by
repeats. Starting with a formal model for representing the repeats, we
show that the number of compressible arcs, which we denote by
$\gamma$, is a relevant parameter for such a
characterization. However, we also prove that, for an arbitrary de
Bruijn graph, identifying a subgraph $G'$ with bounded $\gamma(G')$ is
NP-complete.

\subsection{Simple uniform model for repeats}

We now present the model we adopted for representing the repetition of
a same sequence in a genome or transcriptome. This model is a simple
one and as such should be seen as only a first approximation of what
may happen in reality. Moreover, it allowed us to infer one
characteristic of repeated-associated subgraphs in a de Bruijn graph,
namely that such subgraphs contain few compressible arcs.  It is
important to point out however that the model is realistic enough in
most cases. In particular, it enables to model well recent invasions
of transposable elements which often involve a large number of similar
copies. Such elements are among the most important sources of repeats
in the analysis of NGS data (DNA or RNA-seq).  The number of
compressible arcs is one of the criteria used to define so-called
short nodes in {\sc Oases} but no model or theoretical justification
was given \cite{oases}. This is one of the main contributions of this
paper.  In future, it would be important to consider more realistic
models, and we mention one in the perspectives.  Mathematically, one
has to be aware however that such models will be harder to analyze.

The model is as follows. First, due to mutations, the sequences $s_1,
\ldots, s_m$ that represent the repeats are not identical. However,
provided that the number of such mutations is not high (otherwise the
concept of repeats would not apply), the repeats are considered
``similar'' in the sense of pair-wisely presenting a small Hamming
distance.  We recall that, given two equal length sequences $\vect{s}$
and $\vect{s}'$ in $\Sigma^n$, the \emph{Hamming distance} between
them, denoted by $d_H(\vect{s},\vect{s}')$, is the number of positions
$i$ for which $s[i] \neq s'[i]$.

The model has then the following parameters: $\Sigma$, the length $n$
of the repeat, the number $m$ of copies of the repeat, an integer $k$
(for the length of the $k$-mers considered), and the mutation rate,
$\alpha$, i.e. the probability that a mutation happens in a particular
position.

We first choose uniformly at random a sequence $s_0 \in \Sigma^n$.  At
step $i \leq m$, we create a sequence $s_i$ as follows: for each
position $j$, $s_i[j]=s_0[j]$ with probability $1-\alpha$, whereas
with probability $\alpha$ a value different from $s[j]$ is chosen
uniformly at random for $s_i[j]$. We repeat the whole process $m$
times and thus create a set $S(m,n,\alpha)$ of $m$ such sequences from
$s_0$ (see Fig.\ref{fig:example_repeats} for a small example).  The
generated sequences thus have an expected Hamming distance of $\alpha
n$ from $s_0$.

\begin{figure}[htbp]
\[
\begin{blockarray}{ccccccccccc}
c_1 & c_2 & c_3 & c_4 & c_5 & c_6 & c_7 & c_8 & c_9 & c_{10} \\
\begin{block}{(cccccccccc)c }
A & A & C & T & G & T & A & T & C & C & \quad s_0 \\
A & {\h{C}} & C & T & G & T & A & {\h{G}} & C & C & \quad s_1 \\
{\h{G}} & A & C & T & \h{C} & \h{A} & A & T & C & C & \quad s_2 \\
A & A & C & T & \h{C} & T & A & T & C & C & \quad s_3 \\
A & A & C & A & G & T & A & T & C & \h{A} & \quad s_4 \\
A & A & \h{T} & T & G & T & A & \h{G} & C & C & \quad s_5 \\
A & \h{G} & C & T & G & T & A & \h{T} & C & \h{A} & \quad s_6 \\
\vdots & \vdots & \vdots & \vdots & \vdots & \vdots & \vdots & \vdots & \vdots &\vdots  & \quad  \\
A & A & \h{G} & T & G & \h{A} & A & T & C & C & \quad s_{20} \\
\end{block}
\end{blockarray}
 \]
 \vspace{-2em}
 \caption{\footnotesize{An example of a set of repeats
     $S(20,10,0.1)$.}}
 \label{fig:example_repeats}
\end{figure}
\vspace{-1cm}

\subsection{Topological characterization of the subgraphs generated by repeats}

Given a de Bruijn graph $G_k(R)$, if $a$ is a compressible arc labeled
by the sequence $s=s_1 \ldots s_{k+1}$, then by definition, $a$ is the
only outgoing arc of the vertex labeled by the sequence $s[1,k]$ and
the only incoming arc of the vertex labeled by the sequence
$s[2,k+1]$. Hence the $(k-1)$-mer $s[2,k]$ will appear as a substring
in $R$, always preceded by the symbol $s[1]$ and followed by the
symbol $s[k+1]$. For the sake of simplicity, we refer to such
$(k-1)$-mers as being \emph{boundary rigid}. It is not difficult to
see that the set of compressible arcs in a de Bruijn graph $G_k(R)$
stands in a one-to-one correspondence with the set of boundary rigid
$(k-1)$-mers in $R$.

We now calculate and compare among them the expected number of
compressible arcs in $G=G_k(R)$ when $R$ corresponds to a set of
repeats that are generated: (i) uniformly at random, and (ii)
according to our model.  We show that $\gamma$ is ``small'' in the
cases where the induced graph corresponds to similar sequences, which
provides evidence for the relevance of this parameter.

\begin{claim}
  Let $R$ be a set of $m$ sequences randomly chosen from
  $\Sigma^n$. Then the expected number of compressible arcs in
  $G_k(R)$ is $\Theta(mn)$.
\end{claim}

\begin{proof}
The probability that a sequence of length $k-1$ occurs in a fixed
position in a randomly chosen sequence of length $n$ is
$(1/4)^{k-1}$. Thus the expected number of appearances of a sequence
of length $k-1$ in a set of $m$ randomly chosen sequences of length
$n$ is given by $m(n-k+2) (1/4)^{k-1}$.  If $m(n-k+2) \leq 4^k$, then
this value is upper bounded by $1$, and all the sequences of length
$k-1$ are boundary rigid (as a sequence appears once). The claim
follows by observing that there are $m(n-k+1)$ different
$k$-mers. \qed
\end{proof}

We consider now $\gamma(G_k(R))$ for $R=S(m,n,\alpha)$.  We upper
bound the expected number of compressible arcs by upper bounding the
number of boundary rigid $(k-1)$-mers.

\begin{theorem}\label{theo:expected_compressed_edges}
  Given integers $k, n, m$ with $k<n$ and a real number $0\leq \alpha
  \leq 3/4$, the de Bruijn graph $G_k(S(m,n,\alpha))$ has $o(nm)$
  expected compressible arcs.
\end{theorem}

\begin{proof}
Let $s_0$ be a sequence chosen randomly from $\Sigma^n$.  Let
$S(m,n,\alpha)=\{s_1, \ldots, s_m\}$ be the set of $m$ repeats
generated according to our model starting from $\vect{s}_0$.  Consider
now the de Bruijn graph $G=G_k(S(m,n,\alpha))$. Recall that the number
of compressible arcs in this graph is equal to the number of boundary
rigid $(k-1)$-mers in $S(m,n,\alpha)$. Let $X$ be a random variable
representing the number of boundary rigid $(k-1)$-mers in $G$.
Consider the repeats in $S(m,n,\alpha)$ in a matrix-like ordering as
in Fig.\ref{fig:example_repeats} and observe that the mutations from
one column to another are independent.  Due to the symmetry and the
linearity of expectation, $E[X]$ is given by $m(n-k-1)$ (the total
number of $(k-1)$-mers) multiplied by the probability that a given
$(k-1)$-mer is boundary rigid.

The probability that the $(k-1)$-mer $\hat{s}=s[i,i+k-2]$ is boundary
rigid clearly depends on the distance from the starting sequence
$\hat{s}_0=s_0[i,i+k-2]$.  Let $d$ be the distance $d_H(\hat{s}
,\hat{s}_0)$.

Observe that if the $(k-1)$-mer $s[i] \ldots s[k-1]$ is not boundary
rigid then there exists a sequence $y$ in $S(m,n,\alpha)$ such that
$y[j]= s[j]$ for all $i \leq j \leq i+k-2 $ and either $y[i+k-1] \neq
s[i+k-1]$ or $y[i-1] \neq s[i-1]$. It is not difficult to see that the
probability that this happens is lower bounded by $(2\alpha-4/3
\alpha^2) (1-\alpha)^{k-1-d} (\alpha/3)^d$.  Hence we have:

$$
Pr[\hat{s} \textrm{ is boundary rigid} |  d_H(\hat{s} ,\hat{s}_0)=d ] \leq  \Biggl( 1- (2\alpha-4/3 \alpha^2) (1-\alpha)^{k-1-d} (\alpha/3)^d\Biggr)^{m-1} 
$$

By approximating the above expression we therefore have that,
\begin{flalign}
\DS
E[X]& \leq (n-k-1)m \sum_{d=0}^{k-1} Pr[\hat{s} \textrm{ is boundary rigid} |  d_H(\hat{s} ,\hat{s}_0)=d ] \\ \notag
&\leq (n-k-1) m e^{- (m-1)(2\alpha-4/3 \alpha^2)/(\frac{\alpha}{3})^{k-1}}
\end{flalign}

For a sufficiently large number of copies (\textit{e.g.}
$m=\binom{k}{\alpha k}$) and using the fact that $\binom{k}{\alpha k}
\geq (1/ \alpha )^{\alpha k}$, we have that $E[X]$ is $o(mn)$.  This
concludes the proof. \qed
\end{proof}

The previous result shows that the number of compressible arcs is a
good parameter for characterizing a repeat-associated subgraph.


\subsection{Identifying a repeat-associated subgraph} 

As we showed, a subgraph due to repeated elements has a distinctive
feature: namely it contains few compressible arcs. Based on this, a
natural formulation to the repeat identification problem in RNA-seq
data is to search for large enough subgraphs that do not contain many
compressible arcs. This is formally stated in
Problem~\ref{prob:subgraph}. In order to disregard trivial solutions,
it is necessary to require a large enough \emph{connected} subgraph,
otherwise any set of disconnected vertices or any small subgraph would
be a solution.  Unfortunately, we show that this problem is
NP-complete, so an efficient algorithm for the repeat identification
problem based on this formulation is unlikely.

\begin{problem}[Repeat Subgraph] \label{prob:subgraph}
 
  \textit{INSTANCE:} A directed graph $G$ and two positive integers
  $m$, $t$.

  \textit{DECIDE:} If there exists a (connected) subgraph $G'=(V',
  E')$, with $|V'| \geq m$ and having at most $t$ compressible arcs.
\end{problem}

In Theorem~\ref{thm:np-complete}, we prove that this problem is
NP-complete for all directed graphs with the (total) degree, i.e. the
sum of in and out-degree, bounded by 3. The reduction is from the
Steiner tree problem which requires finding a minimum weight subgraph
spanning a given subset of vertices. It remains NP-hard even when all
arc weights are 1 or 2 (see \cite{Bern89}), this version is denoted by
{\sc STEINER}$(1,2)$. More formally, given a complete undirected graph
$G = (V,E)$ with arc weights in $\{1,2\}$, a set of \emph{terminal}
vertices $N \subseteq V$ and an integer $B$, it is NP-complete to
decide if there exists a subgraph of $G$ spanning $N$ with weight at
most $B$, i.e. a connected subgraph of $G$ containing all vertices of
$N$.

We specify next a family of directed graphs that we use in the
reduction.  Given an integer $x$ we define the directed graph $R(x)$
as a cycle on $2x$ vertices numbered in a clockwise order and where
the arcs have alternating directions, i.e. for any $i \leq x$,
$(v_{2i},v_{2i+1})$ is an arc.  Note that in $R(x)$ all vertices in
even positions, i.e. $v_{2i}$ have out-degree $2$ and in-degree $0$
and those $v_{2i+1}$ out-degree $0$ and in-degree $2$. Clearly, none
of the arcs of $R(x)$ is compressible.

\begin{theorem} \label{thm:np-complete}
  The \emph{Repeat Subgraph Problem} is NP-complete even for directed
  graphs with degree bounded by $d$, for any $d \geq 3$.
\end{theorem}
\begin{proof}
  Given a complete graph $G = (V,E)$, a set of terminal vertices $N$
  and an upper bound $B$, i.e. an instance of {\sc STEINER}$(1,2)$, we
  transform it into an instance of \emph{Repeat Subgraph Problem} with
  a graph $G'$ with degree bounded by $3$. Let us build the graph $G'
  = (V', E')$. For each vertex $v$ in $V \setminus N$, add a
  corresponding subgraph $r(v) = R(|V|)$ in $G'$ and for each vertex
  $v$ in $N$, add a corresponding subgraph $r(v) = R(|E|+|V|^2 + 1)$
  in $G'$. For each arc $(u,v)$ in $E$ with weight $w \in \{1,2\}$,
  add a simple directed path composed by $w$ compressible arcs
  connecting $r(u)$ to $r(v)$ in $G'$, the subgraphs corresponding to
  $u$ and $v$. The first vertex of the path should be in a sink of
  $r(u)$ and the last vertex in a source of $r(v)$. By construction
  there are at least $|V|$ vertices with in-degree $2$ and out-degree
  $0$ (sink) and $|V|$ vertices with out-degree $2$ and in-degree $0$
  (source) in both $r(v)$ and $r(u)$. It is clear that $G'$ has degree
  bounded by $3$. Moreover, the size of $G'$ is polynomial in the size
  of $G$ and it can also be polynomially constructed.

  In this way, the graph $G'$ has one subgraph for each vertex of $G$
  and a path with one or two (depending on the weight of the
  corresponding arc) compressible arcs for each arc of $G$. Thus,
  there exists a subgraph spanning $N$ in $G$ with weight at most $B$
  if and only if there exists a subgraph in $G'$ with at least $m
  =2|N| + 2|E||N| + 2|V|^2|N|$ vertices and at most $t = |B|$
  compressible arcs. This follows from the fact that any subgraph of
  $G'$ with at least $m$ vertices necessarily contains all the
  subgraphs $r(v)$, where $v \in N$, since the number of vertices in
  all $r(v)$, with $v \in V\setminus N$, is at most $|E| + 2|V|^2$ and
  the only compressible arcs of $G'$ are in the paths corresponding to
  the arcs of $G$. \qed
\end{proof}

We can obtain the same result for the specific case of de Bruijn
graphs. The reduction is very similar but uses a different graph family.

\begin{theorem}
  The \emph{Repeat Subgraph Problem} is NP-complete even for subgraphs
  of de Bruijn graphs on $|\Sigma| = 4$ symbols.
\end{theorem}

\section{Bubbles ``drowned'' in repeats}

In the previous section, we showed that an efficient algorithm to
\emph{directly} identify the subgraphs of a de Bruijn graph
corresponding to repeated elements is unlikely to exist since the
problem is NP-complete. However, in this section we show that in the
specific case of a local assembly of alternative splicing (AS) events,
we can \emph{implicitly} avoid such subgraphs. More precisely, it is
possible to find the structures (\textit{i.e.} bubbles) corresponding
to AS events in a de Bruijn graph that are not contained in a repeat
associated subgraph, thus answering to the main open question of
\cite{kissplice}.

\begin{figure}[htbp] 
  \includegraphics[width=\linewidth]{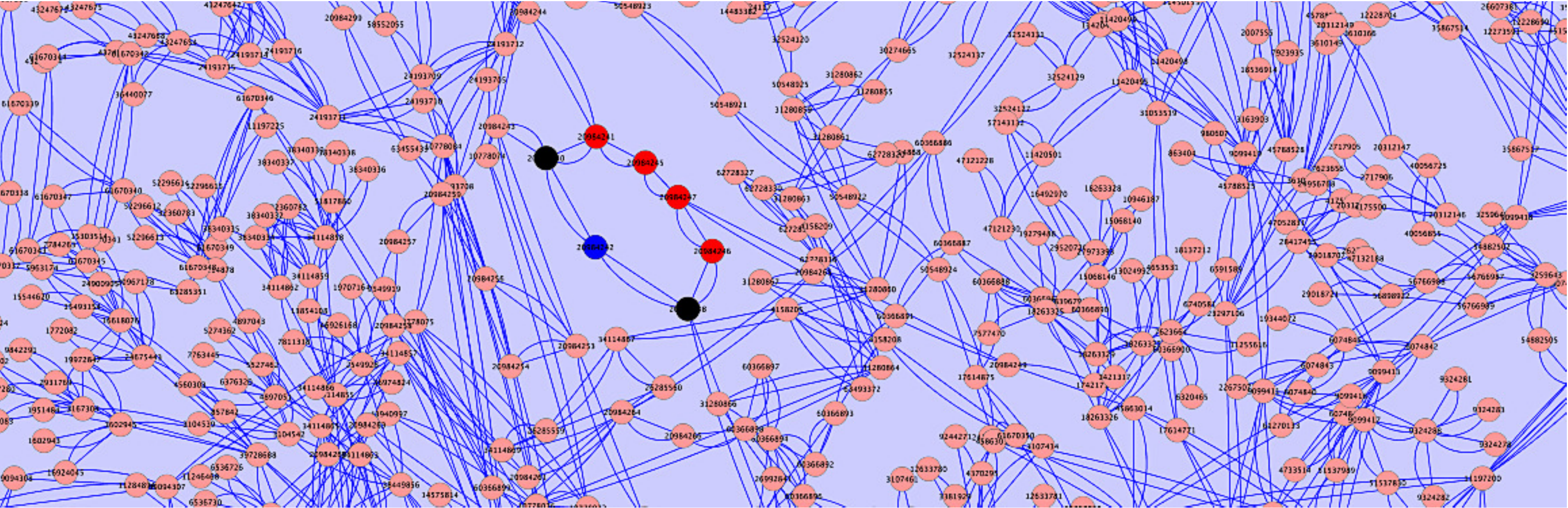}
  \caption{\footnotesize{An alternative splicing event in the SCN5A gene (human) trapped inside a complex
      region, likely containing repeat-associated subgraphs, in a de
      Bruijn graph. The alternative isoforms correspond to a pair of
      paths shown in red and blue.}}
  \label{fig:trapped}
  \vspace{-0.5cm}
\end{figure}

{\sc KisSplice}~\cite{kissplice} is a method for \emph{de novo}
calling of AS events through the enumeration of so-called
\emph{bubbles}, that correspond to pairs of vertex-disjoint paths in a
de Bruijn graph. The bubble enumeration algorithm proposed in
\cite{kissplice} was later improved in \cite{Sacomoto13}. However,
even the improved algorithm is not able to enumerate all bubbles
corresponding to AS events in a de Bruijn graph. There are certain
complex regions in the graph, likely containing repeat-associated
subgraphs but also real AS events~\cite{kissplice}, where both
algorithms take a huge amount of time.  See Fig.~\ref{fig:trapped} for
an example of a complex region with a bubble corresponding to an AS
event. The enumeration is therefore halted after a given timeout. The
bubbles \textit{drowned} (or trapped) inside these regions are thus
missed by {\sc KisSplice}. Here, we impose an extra restriction to the
bubbles, reflecting the fact that they are in complex regions but not
in a repeat-associated subgraph. In this way, we can enumerate bubbles
in complex regions implicitly avoiding repeat-associated subgraphs.

\begin{definition}[$(s,t,\alpha_1,\alpha_2,b)$-bubbles] 
  Given a directed graph $G = (V,E)$ and two vertices $s,t \in V$, an
  $(s,t,\alpha_1,\alpha_2,b)$-bubbles is a pair of vertex-disjoint
  $st$-paths $\pi_1$, $\pi_2$ with lengths bounded by
  $\alpha_1,\alpha_2$, each containing at most $b$ branching vertices.
\end{definition}

This extends $(s,t,\alpha_1,\alpha_2)$-bubbles (Def.~1 in
\cite{Sacomoto13}) by adding the extra condition that each path should
have at most $b$ branching vertices. The number of branching vertices
is proportional to the number of incompressible arcs. Intuitively, if
a bubble contains many incompressible arcs, it is likely contained in
a repeat-associated subgraph. Hence, by limiting the number of
branching vertices (incompressible arcs) we avoid the vertices from
repeat-associated subgraphs. Indeed, in Section~\ref{sec:exp} we show
that by considering bubbles with at most $b$ branching vertices in
{\sc KisSplice}, we increase both sensitivity and precision. This
supports our claim that by focusing on
$(s,t,\alpha_1,\alpha_2,b)$-bubbles, we avoid repeat-associated
subgraphs and recover at least part of the bubbles trapped in complex
regions.

\subsection{Enumerating bubbles avoiding repeats}
In this section, we modify the algorithm of \cite{Sacomoto13} to
enumerate all bubbles with at most $b$ branching vertices in each
path. Given a weighted directed graph $G = (V,E)$ and a vertex $s \in
V$, let $\mathcal{B}_s(G)$ denote the set of
$(s,*,\alpha_1,\alpha_2,b)$-bubbles of $G$. The algorithm recursively
partitions the solution space $\mathcal{B}_s(G)$ at every call until
the considered subspace is a singleton (contains only one solution)
and in that case it outputs the corresponding solution. In order to
avoid unnecessary recursive calls, it maintains the invariant that the
current partition contains at least one solution. It proceeds as
follows.

\smallskip

\noindent{\it Invariant:} At a generic recursive step on vertices
$u_1,u_2$ (initially, $u_1 = u_2 =s$), let $\pi_1 = s \leadsto u_1,
\pi_2 = s \leadsto u_2$ be the paths discovered so far (initially,
$\pi_1,\pi_2$ are empty). Let $G'$ be the current graph (initially,
$G' := G$). More precisely, $G'$ is defined as follows: remove from
$G$ all the vertices in $\pi_1$ and $\pi_2$ but $u_1$ and
$u_2$. Moreover, we also maintain the following invariant ($\ast$):
there exists at least one pair of paths $\bar{\pi}_1$ and
$\bar{\pi}_2$ in $G'$ that extends $\pi_1$ and $\pi_2$ so that $\pi_1
\cdot \bar{\pi}_1$ and $\pi_2 \cdot \bar{\pi}_2$ belongs to
$\mathcal{B}_s(G)$.

\smallskip

\noindent{\it Base case:} When $u_1 = u_2 = u$, output the
$(s,u,\alpha_1,\alpha_2,b)$-bubble given by $\pi_1$ and $\pi_2$.

\smallskip

\noindent{\it Recursive rule:} Let $\mathcal{B}_{s}(\pi_1,\pi_2, G')$
denote the set of $(s,*,\alpha_1,\alpha_2,b)$-bubbles to be listed by
the current recursive call, i.e. the subset of $\mathcal{B}_s(G)$ with
prefixes $\pi_1, \pi_2$. Then, it is the union of the following
disjoint sets\footnote{The same holds for $u_2$ instead of $u_1$.}.
\begin{itemize}
\item The bubbles of $\mathcal{B}_{s}(\pi_1,\pi_2, G')$ that use $e$,
  for each arc $e = (u_1,v)$ out-going from~$u_1$, that is
  $\mathcal{B}_{s}(\pi_1 \cdot e, \pi_2, G' - u_1)$, where $G'-u_1$ is
  the subgraph of $G'$ after the removal of $u_1$ and all its incident
  arcs.
\item The bubbles that do not use any arc from $u_1$, that is
  $\mathcal{B}_{s}(\pi_1,\pi_2, G'')$, where $G''$ is the subgraph of
  $G'$ after the removal of all arcs out-going from $u_1$.
\end{itemize}

In order to maintain the invariant ($\ast$), we only perform the
recursive calls when $\mathcal{B}_{s}(\pi_1 \cdot e, \pi_2, G' - u)$
or $\mathcal{B}_{s}(\pi_1,\pi_2, G'')$ are non-empty. In both cases,
we have to decide if there exists a pair of (internally)
vertex-disjoint paths $\bar{\pi}_1 = u_1 \leadsto t_1$ and
$\bar{\pi}_2 = u_2 \leadsto t_2$, such that $|\bar{\pi}_1| \leq
\alpha_1'$, $|\bar{\pi}_2| \leq \alpha_2'$, and
$\bar{\pi}_1,\bar{\pi}_2$ have at most $b_1,b_2$ branching vertices,
respectively. Since both the length and the number of branching
vertices are monotonic properties, i.e. the length and the number of
branching vertices of a path prefix is smaller than this number for
the full path, we can drop the vertex-disjoint condition. Indeed, let
$\bar{\pi}_1$ and $\bar{\pi}_2$ be a pair of paths satisfying all
conditions but the vertex-disjointness one. The prefixes $\bar{\pi}_1
= u_1 \leadsto t^*$ and $\bar{\pi}_2 = u_2 \leadsto t^*$, where $t^*$
is the first intersection of the paths, satisfy all conditions and are
internally vertex-disjoint. Moreover, using a dynamic programming
algorithm, we can obtain the following result.

\begin{lemma} 
  Given a non-negatively weighted directed graph $G = (V,E)$ and a
  source $s \in V$, we can compute all shortest paths from $s$ using
  at most $b$ branching vertices in $O(b|V||E|)$ time.
\end{lemma}

As a corollary, we can decide if $\mathcal{B}_{s}(\pi_1,\pi_2, G)$ is
non-empty in $O(b|V||E|)$. Now, using an argument similar to
\cite{Sacomoto13}, \textit{i.e.} leaves of the recursion tree and
solutions are in one-to-one correspondence and the height of the
recursion tree is bounded by $2n$, we obtain the following theorem.

\begin{theorem}
  The $(s,*,\alpha_1,\alpha_2,b)$-bubbles can be enumerated in
  $O(b|V|^3|E| |\mathcal{B}_s(G)|)$ time. Moreover, the time elapsed
  between the output of any two consecutive solutions (delay) is
  $O(b|V|^3|E|)$.
\end{theorem}

\subsection{Experimental results} \label{sec:exp}
To evaluate the performance of our method, we simulated RNA-seq data
using the {\sc FluxSimulator} version 1.2.1 \cite{gri12}.  We
generated 100 million reads of 75 bp with the default error model
provided by the {\sc FluxSimulator}. We used the RefSeq annotated
Human transcriptome (hg19 coordinates) as a reference and we performed
a two-step pipeline to obtain a mixture of mRNA and pre-mRNA
(i.e. with introns not yet spliced). To achieve this, we first ran the
{\sc FluxSimulator} with default settings with the Refseq
annotations. Then we modified the annotations to include the introns
and re-ran {\sc FluxSimulator} on this modified version. In this
second run, we additionally constrained the expression values of the
pre-mRNAs to be correlated to the expression values of their
corresponding mRNAs, as simulated in the first run. Finally, we mixed
the two sets of reads to obtain a total of 100M reads. We tested two
values: 5\% and 15\% for the proportion of reads stemming from
pre-mRNAs. Those values were chosen so as to correspond to realistic
ones as observed in a cytoplasmic mRNA extraction (5\%) and a total
(cytoplasmic + nuclear) mRNA extraction \cite{til12}.

On this simulated dataset, we ran {\sc KisSplice} \cite{kissplice}
version 2.1.0 (\textsc{KsOld}) and 2.2.0 (\textsc{KsNew}, with a
maximum number of branching nodes set to 5) and obtained lists of
detected bubbles that are putative alternative splicing (AS) events.

In order to assess the precision and the sensitivity of our method, we
compared our set of {\it found} AS events to the set of {\it true} AS
events.  Following the definition of {\sc Astalavista}, an AS event is
composed of two sets of transcripts, the inclusion/exclusion isoforms
respectively.  An AS event is said to be {\it true} if at least one
transcript among the inclusion isoforms and one among the exclusion
isoforms is present in the simulated dataset with at least one
read. We outline that this definition is very permissive and includes
AS events with very infrequent transcripts.

To compare the results of {\sc KisSplice} with the true AS events, we
propose that a true AS event is {\it found} (counted as a true
positive (TP)) if there is a bubble with one path corresponding to the
inclusion set and the other to the exclusion set. If not, the event is
counted as a false negative (FN). In the meantime, if a bubble does
not correspond to any true AS event, it is counted as a false positive
(FP).  To align the paths of the bubbles to transcript sequences, we
used the {\sc Blat} aligner~\cite{Kent02} with 95\% identity and a
constraint of 95\% of each bubble path length to be aligned (to
account for the sequencing errors simulated by {\sc FluxSimulator}).
We computed the sensitivity TP/(TP+FN) and precision TP/(TP+FP) for
each simulation case and we report their values for various classes of
expression of the minor isoform. Expression values are measured in
reads per kilobase (RPK).

The plots for the sensitivity of each version on the two simulated
datasets are shown in Fig.~\ref{fig:sensitivity}. On the one hand,
both versions of {\sc KisSplice} have similar sensitivity in the 5\%
pre-mRNA dataset, with \textsc{KsNew} performing slightly better,
especially for highly expressed variants. On the other hand, the
sensitivity of the new version is considerably better over all classes
of expression levels in the 15\% pre-mRNA dataset. In this case, the
precisions for \textsc{KsNew} and \textsc{KsOld} are 21\% and 41\%,
respectively. This represents an improvement of almost 100\% over the
old version. The results reflect the fact that the most problematic
repeats are in intronic regions. A small unspliced mRNA rate leads to
few repeat-associated subgraphs, so there are not many AS events
drowned in them (which are then missed by \textsc{KsOld}). In this
case, the advantage of using \textsc{KsNew} is less obvious, whereas a
large proportion of pre-mRNA leads to more AS events drowned in
repeat-associated subgraphs which are identified by \textsc{KsNew} and
missed by \textsc{KsOld}.

Clearly, any improvement in the sensitivity is meaningless if there is
also a significant decrease in precision. This is not the case
here. In both datasets, \textsc{KsNew} \emph{improves} the precision
of \textsc{KsOld}. It increases from 95\% to 98\% and from 90\% to
99\%, in the 5\% and 15\% datasets, respectively. Moreover, both
running times and memory consumption are very similar for both
versions.

In order to give an indication of the usefulness of our
repeat-avoiding bubble enumeration algorithm with real data, we also
ran \textsc{KsNew} and \textsc{KsOld} in the SK-N-SH Human
neuroblastoma cell line RNA-seq dataset (wgEncodeEH000169, total
RNA). In Fig.~\ref{fig:real} we have an example of a
\emph{non-annotated} exon skipping event not found by
\textsc{KsOld}. Observe that the intronic region contains several
transposable elements (many of which are Alu sequences), while the
exons contain none. This is a good example of a bubble (exon skipping
event) drowned in a complex region of the de Bruijn graph. The bubble
(composed by the two alternative paths) itself contains no repeated
elements, but it is surrounded by them. In other words, this is a
bubble with few branching vertices that is surrounded by
repeat-associated subgraphs. Since \textsc{KsOld} is unable to
differentiate between repeat-associated subgraphs and the bubble, it
spends a prohibitive amount of time in the repeat-associated subgraph
and fails to find the bubble.

\begin{figure}[htbp] 
  \includegraphics[height=6cm]{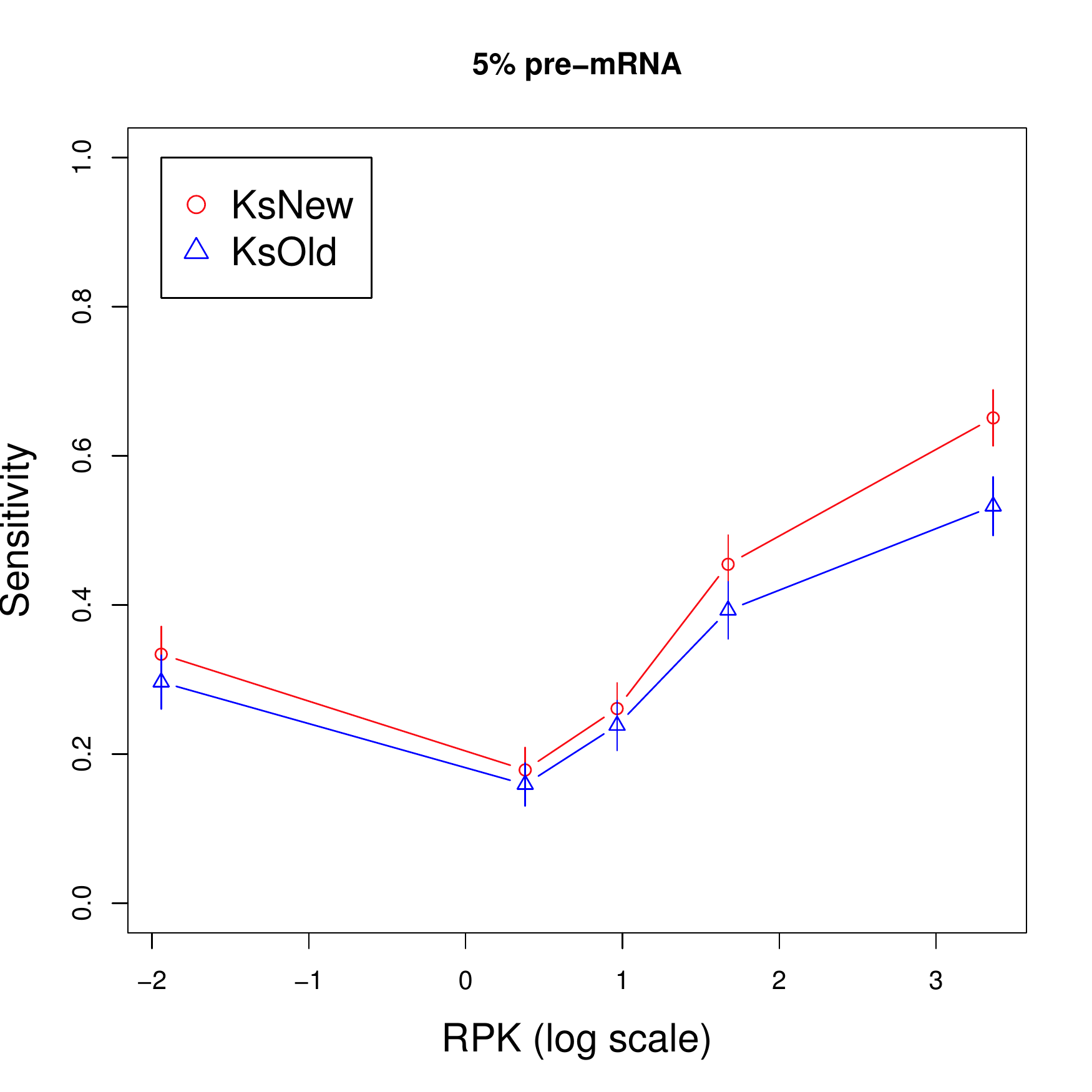}
  \includegraphics[height=6cm]{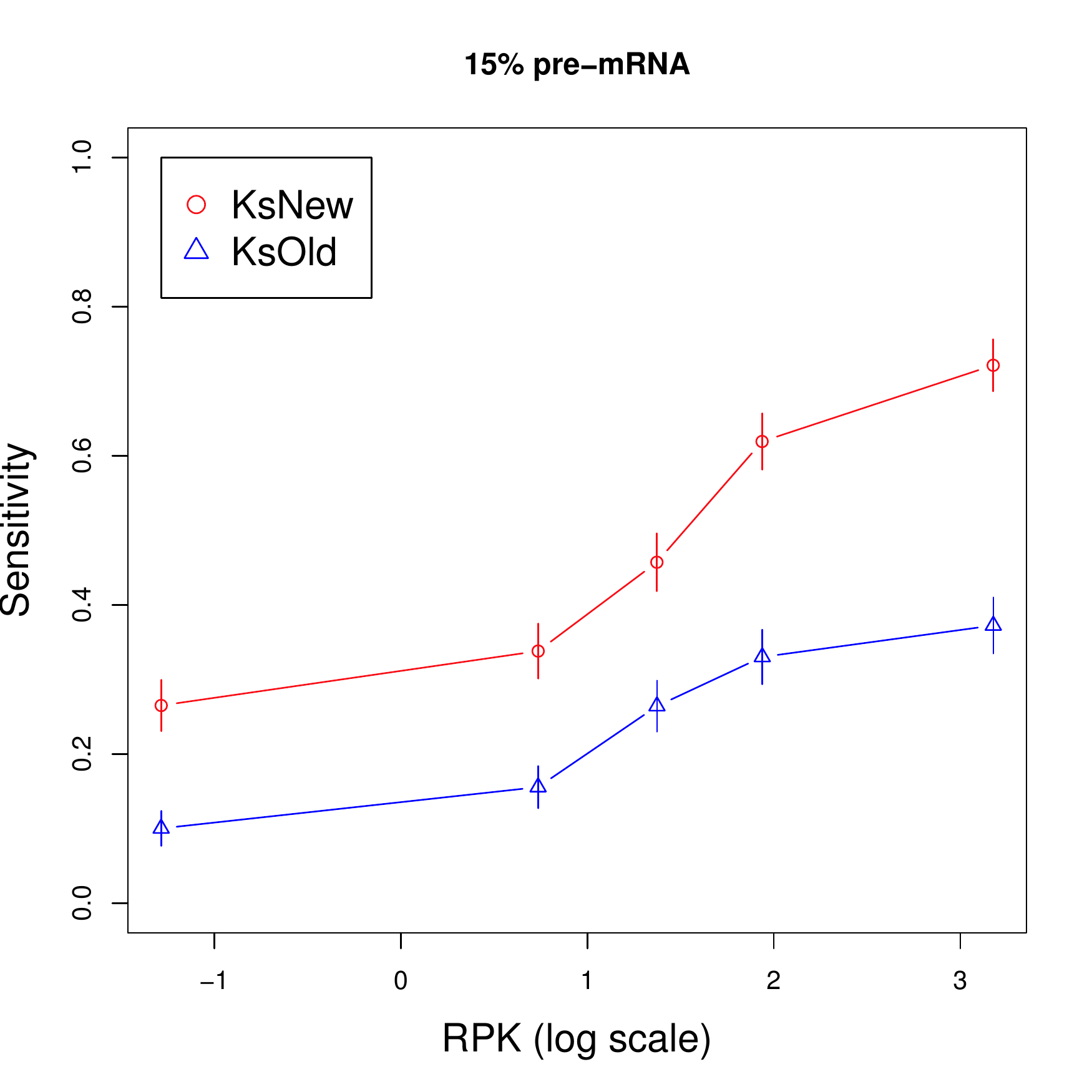}
  \caption{\footnotesize{Sensitivity of \textsc{KsNew} and
      \textsc{KsOld} for several classes of expression of the minor
      isoform. Each class contains the same number of AS events.}}
  \label{fig:sensitivity}
\end{figure}

\begin{figure}[htbp]
  \includegraphics[width=\linewidth]{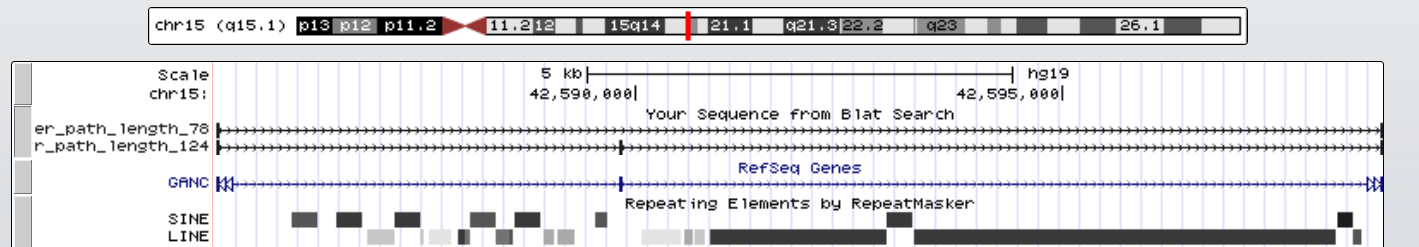}
  \caption{\footnotesize{One of the bubbles found only by
      \textsc{KsNew} with the corresponding sequences mapped to the
      reference human genome and visualized using the UCSC Genome
      Browser. The first two lines correspond to the sequences of,
      respectively, the shortest (exon exclusion variant) and longest
      paths of the bubble mapped to the genome. The blue line is the
      Refseq annotation. The last line shows the annotated SINE and
      LINE sequences (transposable elements). }}
  \label{fig:real}
\end{figure}

\section{Conclusion}
Although transcriptome assemblers are now commonly used, their way to
handle repeats is not satisfactory, arguably because the presence of
repeats in transcriptomes has been underestimated so far. Given that
most RNA-seq datasets correspond to total mRNA extractions, many
introns are still present in the data and their repeat content cannot
be simply ignored. In this paper, we propose a simple model for
repeats. Clearly this model could be improved, for instance by using a
tree-like structure to take into account the evolutionary nature of
repeat (sub)families. Variability in the sizes of the copies of a
repeat family would also enable to model more realistically the true
nature of families of transposable elements (the type of repeats which
cause most trouble in assembly). Certainly, a mathematical analysis of
a more realistic model would be more difficult to obtain. On the other
hand, our simple model captures an important qualitative
characteristic of repeat-associated subgraphs: the presence of few
compressible arcs. This characterization allows us to design an
efficient algorithm to identify bubbles corresponding to AS events
implicitly avoiding repeat-associated subgraphs. This approach
improves both the sensitivity and the precision of \textsc{KisSplice}.

\bibliographystyle{plain} \footnotesize{\bibliography{repeats-rna-seq}}


\end{document}